\documentclass[conference,letterpaper]{IEEEtran}

\IEEEoverridecommandlockouts

\usepackage{cite}
\usepackage{url}
\usepackage{ifthen}
\usepackage{algorithm, algorithmicx, algpseudocode}
\usepackage{graphicx} 
\usepackage{textcomp}
\usepackage{booktabs}
\usepackage[cmex10]{amsmath}
\usepackage{amssymb,amsfonts}
\usepackage{pifont}
\usepackage{amsthm}
\usepackage{mathrsfs}
\def\BibTeX{{\rm B\kern-.05em{\sc i\kern-.025em b}\kern-.08em T\kern-.1667em\lower.7ex\hbox{E}\kern-.125emX}}
\usepackage{tikz}
\usetikzlibrary{automata,positioning,calc}
\usetikzlibrary{intersections}
\usetikzlibrary{decorations.pathreplacing,angles,quotes}

\usepackage{bm}
\usepackage{amscd}

\algnewcommand{\Initialize}[1]{%
  \State \textbf{Initialization:}
  \Statex \hspace*{\algorithmicindent}\parbox[t]{0.8\linewidth}{\raggedright #1}
}
\makeatletter

\theoremstyle{definition}
\newtheorem{theorem}{Theorem}
\newtheorem{definition}{Definition}

\newtheorem{prop}{Proposition}

\newtheorem{remark}{Remark}
\newtheorem{eg}{Example}

\newcommand{\abs}[1]{\left\lvert#1\right\rvert}
\newcommand{\norm}[1]{\left\lVert#1\right\rVert}

\newcommand{\one}[1]{\mbox{1}\hspace{-0.25em}\mbox{l}_{\left\{#1\right\}}}
\newcommand{\argmax}{\operatornamewithlimits{argmax}}
\newcommand{\argmin}{\operatornamewithlimits{argmin}}

\newcommand{\relmiddle}[1]{\mathrel{}\middle#1\mathrel{}} 

\def\br{\mathbb R}

\def\vE{\mathbb E}

\font\b=cmr10 scaled\magstep4

\def\bigzerou{\smash{\lower1.7ex\hbox{\b 0}}}
\def\bigzerou{\smash{\lower1.7ex\hbox{\b 0}}}
 
\begin{document}

\title{A Variational Characterization of $H$-Mutual Information and its Application to Computing $H$-Capacity
\thanks{This work was supported by JSPS KAKENHI Grant Number JP23K16886.}
}

\author{
\IEEEauthorblockN{Akira Kamatsuka}
\IEEEauthorblockA{Shonan Institute of Technology \\ 
Email: \text{kamatsuka@info.shonan-it.ac.jp}
 }
\and
\IEEEauthorblockN{Koki Kazama}
\IEEEauthorblockA{Shonan Institute of Technology \\ 
Email: \text{kazama@info.shonan-it.ac.jp}
 }
\and
\IEEEauthorblockN{Takahiro Yoshida}
\IEEEauthorblockA{Nihon University \\ 
Email: \text{yoshida.takahiro@nihon-u.ac.jp}
 }
}

\maketitle

\begin{abstract}
$H$-mutual information ($H$-MI) is a wide class of information leakage measures, 
where $H=(\eta, F)$ is a pair of monotonically increasing function $\eta$ and a concave function $F$, which is a generalization of Shannon entropy.
$H$-MI is defined as the difference between the generalized entropy $H$ and its conditional version, 
including Shannon mutual information (MI), Arimoto MI of order $\alpha$, $g$-leakage, and expected value of sample information. 
This study presents a variational characterization of $H$-MI via statistical decision theory. 
Based on the characterization, we propose an alternating optimization algorithm for computing $H$-capacity.

\end{abstract}

\begin{IEEEkeywords}
$H$-mutual information, Arimoto--Blahut algorithm, statistical decision theory, value of information
\end{IEEEkeywords}

\section{Introduction}
Shannon mutual information (MI) $I(X; Y)$ \cite{shannon} is a typical quantity that quantifies 
the amount of information a random variable $Y$ contains about a random variable $X$. 
Several ways to generalize the Shannon MI are available in literature. 
A well-known generalization of Shannon MI is a class of $\alpha$-mutual information ($\alpha$-MI) $I_{\alpha}^{(\cdot)}(X; Y)$ \cite{7308959}, 
where $\alpha\in (0, 1)\cup (1, \infty)$ is a tunable parameter. 
The $\alpha$-MI class includes Sibson MI $I_{\alpha}^{\text{S}}(X; Y)$ \cite{Sibson1969InformationR}, 
Arimoto MI $I_{\alpha}^{\text{A}}(X; Y)$ \cite{arimoto1977}, and Csisz\'{a}r MI $I_{\alpha}^{\text{C}}(X; Y)$ \cite{370121}. 
These MIs share common properties such as non-negativity and data-processing inequality (DPI).

In problems on information security, Shannon MI can be interpreted as a measure of information leakage, 
i.e., a measure of how much information observed data $Y$ leak about secret data $X$.  
Recently, various operationally meaningful leakage measures were proposed for privacy-guaranteed data-publishing problems. 
For example, Calmon and Fawaz introduced the \textit{average cost gain} \cite{6483382} and Issa \textit{et al.} introduced the \textit{maximal leakage}.  
Extending the maximal leakage, Liao \textit{et al.} introduced \textit{$\alpha$-leakage} and \textit{maximal $\alpha$-leakage} \cite{8804205}. 
Alvim \textit{et al.} proposed \textit{$g$-leakage} \cite{6957119,ALVIM201932,6266165}, a rich class of information leakage measures; 
$g$-leakage was extended to \textit{maximal $g$-leakage} by Kurri \textit{et al.} \cite{10352344}.
Note that these information leakage measures are based on the adversary's decision-making on $X$ from the observed data $Y$ and a gain (utility) or loss (cost) function. 

Research on quantifying leaked information from the observed data $Y$ based on a decision-making problem can be traced back to the 1960s.
In a pioneering work by Raiffa and Schlaifer on quantifying the \textit{value of information} (VoI) \cite{raiffa1961applied}, 
the \textit{expected value of sample information} (EVSI) was formulated in a statistical decision-theoretic framework.
EVSI was defined as the largest increase in maximal Bayes expected gain (or the largest reduction of minimal Bayes risk) compared to those without using $Y$. 
Thus, information leakage measures in the information disclosure problem can be interpreted as variants of EVSI.

Recently, Am{\'e}rico \textit{et al.} proposed a wide class of information leakage measures, referred to as $H$-mutual information ($H$-MI) $I_{H}(X; Y)$ \cite{9505206,9064819}. 
Here, $H=(\eta, F)$ is a pair of a continuous real-valued function $F\colon \Delta_{\mathcal{X}}\to \br$ and a continuous and strictly increasing function $\eta\colon F(\Delta_{\mathcal{X}}) \to \br$, 
where $\Delta_{\mathcal{X}}$ is a probability simplex on a finite set $\mathcal{X}$ and $F(\Delta_{\mathcal{X}})$ is the image of $F$. 
When $\eta$ is an identity map and $F(p_{X}):=-\sum_{x}p_{X}(x)\log p_{X}(x)$, $H=(\eta, F)$ represents the Shannon entropy {$S(X)$}. Thus $H=(\eta, F)$ can be regarded as a generalized entropy.
$H$-MI is defined as the difference between the generalized entropy $H=(\eta, F)$ and its conditional version $H(X | Y)$, 
which includes Shannon MI, Arimoto MI of order $\alpha$, $g$-leakage, and EVSI. 
In \cite{9505206,9064819}, Am{\'e}rico \textit{et al.} provided the necessary and sufficient conditions (referred to as \textit{core-concavity} (CCV) condition) for $I_{H}(X; Y)$ to satisfy 
non-negativity and DPI when the conditional entropy $H(X|Y)$ satisfies the \textit{$\eta$-averaging} (EAVG) condition. 

In this study, we present a variational characterization of $H$-MI that satisfies DPI via statistical decision theory.
Our variational characterization transforms $H$-MI into the following optimization problem: 
\begin{align}
I_{H}(X; Y) &= \max_{q_{X\mid Y}} {\mathcal{F}_{H}}(p_{X}, q_{X\mid Y}), 
\end{align}
where $p_{X}\in \Delta_{\mathcal{X}}$ is a distribution on $X$ and $q_{X\mid Y}=\{q_{X\mid Y}(\cdot\mid y)\}_{y\in \mathcal{Y}}$ is a set of conditional distributions of $X$, given $Y=y$. 
This variational characterization allows us to derive an alternating optimization algorithm (also known as Arimoto--Blahut algorithm \cite{1054753}, \cite{1054855}) for computing $H$-capacity $C_{H}:=\max_{p_{X}}I_{H}(X; Y)$, 
such as the channel capacity $C:=\max_{p_{X}} I(X; Y)$ and Arimoto capacity $C_{\alpha}^{\text{A}}:=\max_{p_{X}}I_{\alpha}^{\text{A}}(X; Y)$
\footnote{It is worth mentioning that Liao \textit{et al.} reported the operational meaning of Arimoto capacity and Sibson capacity in the privacy-guaranteed data-publishing problems \cite[Thm 2]{8804205}; these capacities are essentially equivalent to the maximal $\alpha$-leakage.} \cite{arimoto1977,BN01990060en}, \cite{kamatsuka2024new}.  

\subsection{Main Contributions}
The main contributions of this study are as follows:

\begin{itemize}
\item We provide a variational characterization of $H$-MI  (Theorem \ref{thm:variational_characterization})
using the fact  that every concave function $F$ has a statistical decision-theoretic variational characterization \cite[Section 3.5.4]{grunwald2004}. 
\item On the basis of variational characterization, we build an alternating optimization algorithm for calculating $H$-capacity $C_{H}:=\max_{p_{X}}I_{H}(X; Y) = \max_{p_{X}}\max_{q_{X\mid Y}}{\mathcal{F}_{H}}(p_{X}, q_{X\mid Y})$ (Algorithm \ref{alg:ABA_H}) (see Section \ref{sec:application_ABA}). 
Moreover, we show that the algorithms for computing Arimoto capacity $C_{\alpha}^{\text{A}}$ derived from our approach coincide with the previous algorithms reported in \cite{BN01990060en}, \cite{kamatsuka2024new}. 
\end{itemize}

\subsection{Organization of the Paper}
The remainder of this paper is organized as follows. 
We review the statistical decision theory and $H$-MI in Section \ref{sec:preliminaries}. 
In Section \ref{sec:variational_characterization_H-MI}, we present the variational characterization of $H$-MI. 
In Section \ref{sec:application_ABA}, we derive an alternating optimization algorithm for computing $H$-capacity $C_{H}:=\max_{p_{X}}I_{H}(X; Y)$ based on the characterization.

\section{Preliminaries}\label{sec:preliminaries}

\subsection{Notations}
Let $X, Y$ be random variables on finite alphabets $\mathcal{X}$ and $\mathcal{Y}$, 
drawn according to a joint distribution $p_{X, Y} = p_{X}p_{Y\mid X}$. 
Let $p_{Y}$ be a marginal distribution of $Y$ and 
$p_{X\mid Y}(\cdot | y):=\frac{p_{X}(\cdot)p_{Y\mid X}(y\mid \cdot)}{\sum_{x}p_{X}(x)p_{Y\mid X}(y\mid x)}$  be a posterior distribution on $X$ given $Y=y$, respectively. 
The set of all distributions $p_{X}$ is denoted as $\Delta_{\mathcal{X}}$. We often identify $\Delta_{\mathcal{X}}$ with $(m-1)$-dimensional probability simplex 
$\left\{\, (p_{1},\dots, p_{m})\in [0, 1]^{m} \relmiddle{|} \sum_{i=1}^{m} p_{i}=1 \right\}$, where $m:=\abs{\mathcal{X}}$.
Given a function $f\colon \mathcal{X}\to \br$,  we use $\vE_{X}[f(X)]:=\sum_{x}f(x)p_{X}(x)$ and $\vE_{X}[f(X) | Y=y]:=\sum_{x}f(x)p_{X\mid Y}(x | y)$ to denote expectation on $f(X)$ and 
conditional expectation on $f(X)$ given $Y=y$, respectively.
We also use $\vE_{X}^{p_{X}}[f(X)]$ to emphasize that we are taking expectations $p_{X}$.
We use {$S(X)$, $S(X | Y)$, $I(X;Y):=S(X)-S(X | Y)$}\footnote{Note that, throughout this paper, the notations $H(X)$ and $H(X|Y)$ are used to denote generalized forms of entropy and conditional entropy introduced in Definitions \ref{def:generalized_entropy} and \ref{def:EAVG}.}, and $D(p || q )$ to denote Shannon entropy, conditional entropy, Shannon MI, and relative entropy, respectively. 
Let $\mathcal{A}$ be an action space (decision space) and $\delta\colon \mathcal{Y}\to \mathcal{A}$ be a decision rule for a decision maker (DM). 
Let $A := \delta(Y)$ be an action (decision) of the DM. We use $\ell(x, a)$ and $g(x, a)$ to denote the loss (cost) function and  gain (utility) function of the DM, respectively. 
Throughout this paper, we use $\log$ to denote the natural logarithm and $\norm{p_{X}}_{p}:=(\sum_{x}p_{X}(x)^{p})^{\frac{1}{p}}$ represents the $p$-norm of $p_{X}\in \Delta_{\mathcal{X}}$. 

We initially review statistical decision theory \cite{GVK027440176} and $H$-MI \cite{9505206,9064819}. 

\subsection{Statistical Decision Theory and Scoring Rules} \label{ssec:sdt}
In this subsection, we review statistical decision theory. 
In particular, we review a problem of deciding the optimal probability mass function (pmf) considering a loss or a gain function (referred to as a \textit{scoring rule}), 
which is historically known as a \textit{probability forecasting} problem.

Suppose that a DM makes action $A\in \mathcal{A}$ from observed data $Y\in \mathcal{Y}$ using a decision rule $\delta\colon \mathcal{Y}\to \mathcal{A}$ . 
We assume that the DM uses the decision rule $\delta^{*}$ that minimizes Bayes risk $r(\delta) := \vE_{X, Y}\left[\ell(X, \delta(Y))\right]$ (or maximizes Bayes expected gain $G(\delta):=\vE_{X, Y}\left[g(X, \delta(Y))\right]$). 
Figure \ref{fig:system_model} shows the system model for this problem.

\begin{figure}[htbp]
\centering
\resizebox{0.5\textwidth}{!}{
\begin{tikzpicture}[auto]
\tikzset{block/.style={draw, rectangle, minimum height = .5cm, minimum width = .5cm, text centered, thick}};
\node (database) {$X$}; 
\node[block,right=.5cm of database] (mechanism) {$p_{Y\mid X}$};
\node[right=.5cm of mechanism] (output) {$Y$};
\node[block, right=.5cm of output] (d_bob) {$\delta$};
\node[right=.5cm of d_bob] (T_hat) {$A$};

\draw[->, thick, >=latex] (database) --  (mechanism);
\draw[->, thick, >=latex] (mechanism) -- (output);
\draw[->, thick, >=latex] (output) -- (d_bob);
\draw[->, thick, >=latex] (d_bob) -- (T_hat);
\end{tikzpicture}
}
\caption{System model of the statistical decision theory}
\label{fig:system_model}
\end{figure}
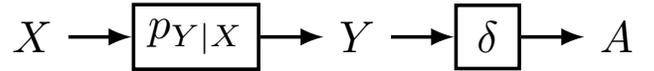

\begin{prop}[\text{\cite[Result 1]{GVK027440176}, \cite[Thm 2.7]{alma991009366059705251}}] \label{prop:optimal_decision_rule}
The minimal Bayes risk is given by 
\begin{align}
\min_{\delta}r(\delta) &= r(\delta^{\text{*}}) \\ 
&= \vE_{Y}\left[\min_{a\in \mathcal{A}} \vE_{X}\left[\ell(X, a)\relmiddle{|} Y \right]\right]  \label{eq:min_posterior_loss} \\ 
&= \sum_{y}p_{Y}(y) \left[\min_{a\in \mathcal{A}}\sum_{x}\ell(x,a) p_{X\mid Y}(x\mid y)\right], 
\end{align}
with the optimal decision rule $\delta^{\text{*}}\colon \mathcal{Y}\to \mathcal{A}$ given by 
\begin{align}
\delta^{\text{*}}(y) &:= \argmin_{a\in \mathcal{A}} \vE_{X}\left[\ell(X, a)\mid Y=y\right]. 
\end{align}
Similarly, the maximal Bayes expected gain and the optimal decision rule $\delta^{\text{*}}\colon \mathcal{Y}\to \mathcal{A}$ are given by 
\begin{align}
\max_{\delta} G(\delta) &= G(\delta^{*}) \\ 
&= \vE_{Y}\left[\max_{a\in \mathcal{A}} \vE_{X}\left[g(X, a)\relmiddle{|} Y\right]\right], \\ 
\delta^{\text{*}}(y) &:= \argmax_{a\in \mathcal{A}} \vE_{X}\left[g(X, a)\mid Y=y\right]. \label{eq:max_posterior_utility}
\end{align}
\end{prop}

\begin{remark} \label{rem:affine_same_optimal}
Let $\ell(x, a)$ be a loss function. Let us define a gain function $g(x, a):=c\ell(x, a) + d$, where $c<0$ and $d$ are constants. 
One can easily see that if $\delta^{*}$ minimize Bayes risk $r(\delta):=\vE_{X, Y}\left[\ell(X, \delta(Y))\right]$ 
then the rule $\delta^{*}$ also maximizes the Bayes expected gain $G(\delta):=\vE_{X, Y}\left[g(X, \delta(Y))\right]$. 
The reverse is also true.
\end{remark}

\begin{eg} \label{eg:0-1_loss}
{Let $\hat{X}$ be an estimator of $X$.} Suppose that a DM conducts a point estimation on $X$, i.e., $A=\hat{X}\in \mathcal{X}$ considering $0$-$1$ loss  $\ell_{\text{$0$-$1$}}(x, \hat{x})=\one{x=\hat{x}}, $ where $\one{\cdot}$ is {an} indicator function. 
Then the minimal Bayes risk and the optimal decision rule $\delta^{*}$ are given as follows:
\begin{align}
\min_{\delta}r(\delta) &= 1-\vE_{Y}\left[\max_{x} p_{X\mid Y}(x\mid Y)\right], \\ 
\delta^{*}(y) &= \argmax_{x} p_{X\mid Y}(x\mid y). \qquad \text{(MAP estimation)}
\end{align}
\end{eg}

\begin{eg} \label{eg:log_score}
Suppose that a DM decides the optimal pmf $q\in \mathcal{A}=\Delta_{\mathcal{X}}$ considering log-score $g_{\text{log}}(x, q):= \log q(x)$ \cite{https://doi.org/10.1111/j.2517-6161.1952.tb00104.x}.
Then, the maximal Bayes expected gain and the optimal decision rule are given as 
\begin{align}
\min_{\delta}r(\delta) &= {S(X\mid Y)}, \\ 
\delta^{*}(y) &= p_{X\mid Y}(\cdot\mid y), 
\end{align}
where  ${S(X | Y)}=-\sum_{y}p_{Y}(y)\sum_{x}p_{X\mid Y}(x | y)\log p_{X\mid Y}(x | y)$ is the conditional entropy.
\end{eg}

\begin{remark}
Historically, the problem of deciding the optimal pmf $q\in \Delta_{\mathcal{X}}$ considering a loss $\ell(x, q)$ or a gain $g(x, q)$ 
is called a \textit{probability forecasting} problem \cite{doi:10.1198/016214506000001437}, \cite{Dawid:2014ua}. 
In the problem, the loss or gain function is called the \textit{scoring rule}.
\end{remark}

\begin{remark} \label{rem:q_X_Y_risk}
Note that finding the optimal decision rule $\delta\colon \mathcal{Y}\to \Delta_{\mathcal{X}}$ that minimizes $r(\delta)$ (\textit{resp.} maximizes $G(\delta)$) is 
equivalent to finding the optimal set of conditional distributions $q_{X\mid Y}=\{q_{X\mid Y}(\cdot \mid y)\}_{y\in \mathcal{Y}}$ that 
minimizes $r(q_{X\mid Y}):=\vE_{X, Y}\left[\ell(X, q_{X\mid Y}(X \mid Y))\right]$ (\textit{resp.} maximizes $G(q_{X\mid Y}):=\vE_{X, Y}\left[g(X, q_{X\mid Y}(X \mid Y))\right]$). 
Thus we call $r(q_{X\mid Y})$ (\textit{resp.} $G(q_{X\mid Y})$) as  Bayes risk (\textit{resp.}  Bayes expected gain) for $q_{X\mid Y}$ and 
denote the optimal set of conditional distribution as $q_{X\mid Y}^{*}$. 
\end{remark}

\begin{eg} Besides the log-score $g_{\text{log}}(x, q)$ in Example \ref{eg:log_score}, there exist other scoring rules that give the same optimal set of conditional distribution $q_{X\mid Y}^{*}$.
Some examples are shown below: 
\begin{itemize}
\item $g_{\text{PS}}(x, q) := \frac{1}{\alpha-1} \left(\frac{q(x)}{\norm{q}_{\alpha}} \right)^{\alpha-1}$ (the \textit{pseudo-spherical score} \cite{Good_1971})
\item $g_{\text{Power}}(x, q) := \frac{\alpha}{\alpha-1}\cdot q(x)^{\alpha-1} - \norm{q}_{\alpha}^{\alpha}$ (the \textit{power score} \cite{Selten:1998aa} (also known as \textit{Tsallis score} \cite{Dawid:2014ua})) 
\end{itemize}
\footnote{The pseudo-spherical score and the power score are originally defined for $\alpha>1$. We multiply the original definitions by $\frac{1}{\alpha-1}$ so that we can define them for $\alpha\in (0, 1) \cup (1, \infty)$.}
\end{eg}

Note that the log-score $g_{\log}(x, q)$, pseudo-spherical score $g_{\text{PS}}(x, q)$, and power score $g_{\text{Power}}(x, q)$ 
are all \textit{proper scoring rules} (PSR) defined as follows.

\begin{definition} \label{def:PSR}
The scoring rule $g(x, q)$ is \textit{proper} if for all $q\in \Delta_{\mathcal{X}}$, 
\begin{align}
\vE_{X}^{p_{X}}\left[g(X, p_{X})\right] \geq \vE_{X}^{p_{X}}\left[g(X, q)\right]. \label{eq:PSR}
\end{align} 
If the equality holds if and only if $q=p_{X}$, then the scoring rule $g(x, q)$ is called \textit{strictly proper}\footnote{Similarly, we can define a (strictly) proper loss $\ell(x, q)$. }. 
\end{definition}

\begin{eg} 
Recently, Liao \textit{et al.} proposed \textit{$\alpha$-loss} $\ell_{\alpha}(x, q):=\frac{\alpha}{\alpha-1}\left( 1-q(x)^{\frac{\alpha-1}{\alpha}} \right)$ \cite[Def 3]{8804205} in the privacy-guaranteed data-publishing context. In \cite[Lemma 1]{8804205}, they proved that 
\begin{align}
\argmin_{q}\vE_{X}^{p_{X}}\left[\ell_{\alpha}(X, q)\right] = p_{X_{\alpha}}, 
\end{align}
where $p_{X_{\alpha}}$ is the \textit{$\alpha$-tilted distribution of $p_{X}$} (also known as \textit{scaled distribution \cite{7308959} and escort distribution \cite{10.5555/3019383}}) defined as follows:
\begin{align}
p_{X_{\alpha}}(x) &:= \frac{p_{X}(x)^{\alpha}}{\sum_{x} p_{X}(x)^{\alpha}}. \label{eq:alpha_tilted_dist}
\end{align}
Thus, $\alpha$-loss $\ell_{\alpha}(x, q)$ can be regard as a scoring rule that is \textit{not} proper. 
\end{eg}

Table \ref{tab:tab_scoring_rule} summarizes examples of scoring rules described above, their optimal values, and the optimal set of conditional distributions $q_{X\mid Y}^{*}$.

\setlength\tabcolsep{1.pt}
\begin{table*}[ht]
  \caption{Typical scoring rules for deciding $q\in \Delta_{\mathcal{X}}$ and the optimal decision rules}
  \label{tab:tab_scoring_rule}
  \centering
  \resizebox{1.0\textwidth}{!}{
  {
  \begin{tabular}{@{} ccccc @{}}
    \toprule
    \begin{tabular}{c}
    $\ell(x, q)$, \\ $g(x, q)$ 
    \end{tabular}
    & \begin{tabular}{c}
    $\argmin_{q}\vE_{X}\left[\ell(X, q)\right]$ \\ 
    $=\argmax_{q}\vE_{X}\left[g(X, q)\right]$
    \end{tabular} 
    & \begin{tabular}{c}
    $\min_{q}\vE_{X}\left[\ell(X, q)\right]$, \\ 
    $\max_{q}\vE_{X}\left[g(X, q)\right]$
    \end{tabular} 
    & \begin{tabular}{c}
    $\argmin_{q_{X\mid Y}}\vE_{X, Y}\left[\ell(X, q_{X\mid Y}(\cdot | Y))\right]$ \\ 
    $=\argmax_{q_{X\mid Y}}\vE_{X, Y}\left[g(X, q_{X\mid Y}(\cdot | Y))\right]$
    \end{tabular} 
    & \begin{tabular}{c}
    $\min_{q_{X\mid Y}}\vE_{X, Y}\left[\ell(X, q_{X\mid Y}(\cdot | Y))\right]$, \\ 
    $\max_{q_{X\mid Y}}\vE_{X, Y}\left[g(X, q_{X\mid Y}(\cdot | Y))\right]$ 
    \end{tabular} \\ 
    \midrule
    \begin{tabular}{c}
    $-\log q(x)$ (log-loss), \\ 
    $\log q(x)$ (log-score \cite{https://doi.org/10.1111/j.2517-6161.1952.tb00104.x}) 
    \end{tabular}
    & $p_{X}$ 
    & \begin{tabular}{c}
    $S(X)$, \\ 
    $-S(X)$
    \end{tabular} 
    & $p_{X\mid Y}(\cdot | y), y\in \mathcal{Y}$ 
    & \begin{tabular}{c}
    $S(X | Y)$, \\
    $-S(X | Y)$
    \end{tabular} \\ 
    \begin{tabular}{c}
    $\frac{1}{\alpha-1}\left(1- \left(  \frac{q(x)}{\norm{q}_{\alpha}} \right)^{\alpha-1} \right)$, \\
    $\frac{1}{\alpha-1}\cdot \left(\frac{q(x)}{\norm{q}_{\alpha}} \right)^{\alpha-1}$ \\ 
    (pseudo-spherical score \cite{Good_1971}) 
    \end{tabular}
    & $p_{X}$ 
    & \begin{tabular}{c}
    $\frac{1}{\alpha-1}\left( 1-\norm{p_{X}}_{\alpha} \right)$ \\ 
    (Harvda--Tsallis entropy), \\ 
    $\frac{1}{\alpha-1}\cdot \norm{p_{X}}_{\alpha}$
    \end{tabular} 
    & $p_{X\mid Y}(\cdot | y), y\in \mathcal{Y}$ 
    & \begin{tabular}{c}
    $\frac{1}{\alpha-1}\left( 1-\vE_{Y}\left[\norm{p_{X\mid Y}(\cdot | Y)}_{\alpha}\right] \right)$, \\ 
    $\frac{1}{\alpha-1}\cdot\vE_{Y}\left[\norm{p_{X\mid Y}(\cdot | Y)}_{\alpha}\right]$
    \end{tabular} \\ 
    \begin{tabular}{c}
    $\frac{\alpha}{\alpha-1}\left( 1- q(x)^{\alpha-1} \right) + \norm{q}_{\alpha}^{\alpha}$, \\
    $\frac{\alpha}{\alpha-1}\cdot q(x)^{\alpha-1} - \norm{q}_{\alpha}^{\alpha}$ \\ 
    \text{(power score \cite{Selten:1998aa}}, \\ 
    \text{Tsallis score \cite{Dawid:2014ua})}
    \end{tabular}
    & $p_{X}$ 
    & \begin{tabular}{c}
    $\frac{1}{\alpha-1}(1-\norm{p_{X}}_{\alpha}^{\alpha})$, \\ 
    $\frac{1}{\alpha-1}\cdot \norm{p_{X}}_{\alpha}^{\alpha}$
    \end{tabular} 
    & $p_{X\mid Y}(\cdot | y), y\in \mathcal{Y}$ 
    & \begin{tabular}{c}
    $\frac{1}{\alpha-1}\left( 1-\vE_{Y}\left[\norm{p_{X\mid Y}(\cdot | Y)}_{\alpha}^{\alpha}\right] \right)$, \\ 
    $\frac{1}{\alpha-1}\cdot\vE_{Y}\left[\norm{p_{X\mid Y}(\cdot | Y)}_{\alpha}^{\alpha}\right]$
    \end{tabular}  \\
    \begin{tabular}{c}
    $\frac{\alpha}{\alpha-1}\left( 1-q(x)^{\frac{\alpha-1}{\alpha}} \right)$ \text{($\alpha$-loss \cite{8804205})}, \\ 
    $\frac{\alpha}{\alpha-1}\cdot q(x)^{\frac{\alpha-1}{\alpha}}$ \text{($\alpha$-score)}
    \end{tabular}
    & $p_{X_{\alpha}}$ 
    & \begin{tabular}{c}
    $\frac{\alpha}{\alpha-1}\left(1-\norm{p_{X}}_{\alpha} \right)$, \\ 
    $\frac{\alpha}{\alpha-1}\cdot \norm{p_{X}}_{\alpha}$ 
    \end{tabular} 
    & $p_{X_{\alpha}\mid Y}(\cdot | y), y\in \mathcal{Y}$ 
    & \begin{tabular}{c}
    $\frac{\alpha}{\alpha-1} \left( 1-\vE_{Y}\left[\norm{p_{X\mid Y}(\cdot\mid Y)}_{\alpha}\right] \right)$, \\
    $\frac{\alpha}{\alpha-1} \cdot \vE_{Y}\left[\norm{p_{X\mid Y}(\cdot | Y)}_{\alpha}\right]$
    \end{tabular} \\ 
    \bottomrule
  \end{tabular}
  }
  }
\end{table*}

\subsection{$H$-Mutual information ($H$-MI) \cite{9505206,9064819}}
In this subsection, we review $H$-MI and show that $H$-MI includes well-known information leakage measures.

\begin{definition}[\text{\cite[Def. 11]{9505206}}] \label{def:generalized_entropy}
Let $p_{X}$ be a pmf of $X$, $F\colon \Delta_{\mathcal{X}}\to \br$ and $\eta\colon F(\Delta_{\mathcal{X}})\to \br$ be continuous functions,  
and $\eta$ be strictly increasing. 
Given $H=(\eta, F)$, the \textit{unconditional form of entropy} is defined as follows:
\begin{align}
H(X) := \eta(F(p_{X})).
\end{align}
\end{definition}

\begin{definition}[\text{CCV \cite[Def. 12]{9505206}}] \ 
$H=(\eta, F)$ is \textit{core-concave} (CCV) if $F$ is concave. 
We say that $H(X)$ is core-concave entropy if $H=(\eta, F)$ is CCV.
\end{definition}

\begin{definition}[\text{EAVG \cite[Def. 13]{9505206}}] \label{def:EAVG}
\footnote{We slightly modified the definition of \text{EAVG}.}
Given a joint distribution $p_{X, Y}=p_{X}p_{Y\mid X}$ and $H=(\eta, F)$, 
a functional $H(p_{X}, p_{Y\mid X})$ satisfies \textit{$\eta$-averaging} (EAVG) if it is represented as follows:
\begin{align}
H(p_{X}, p_{Y\mid X}) &= \eta \left( \vE_{Y}^{p_{Y}}\left[F(p_{X\mid Y}(\cdot \mid Y))\right] \right) \\ 
&= \eta \left( \sum_{y} p_{Y}(y)F(p_{X\mid Y}(\cdot \mid y)) \right) {, }
\end{align}
{where $p_{X\mid Y}(x|y):= \frac{p_{X}(x)p_{Y\mid X}(y | x)}{\sum_{x}p_{X}(x)p_{Y\mid X}(y | x)}$ is the posterior distribution of $X$ given $Y=y$ and $p_{Y}(y):=\sum_{x}p_{X}(x)p_{Y\mid X}(y|x)$ is the marginal distribution of $Y$. }
We say that $H(p_{X}, p_{Y\mid X})$ is conditional entropy of $H=(\eta, F)$ and it is denoted by $H(X | Y)$. 
\end{definition}

\begin{theorem}[\text{\cite[Thm. 2]{9064819} and \cite[Thm. 4]{9505206}}] \label{thm:iff_DPI}
Given $H=(\eta, F)$, $H$-MI is defined as 
\begin{align}
I_{H}(X; Y) &:= H(X) - H(X\mid Y), 
\end{align}
where $H(X | Y)$ satisfies EAVG. Then, the following are equivalent\footnote{Note that the original statement of the theorem is stated in terms of conditional entropy $H(X | Y)$ instead of $H$-MI $I_{H}(X; Y)$.}: 
\begin{description}
\item[(CCV)] $H=(\eta, F)$ is core-concave.
\item[(Non-negativity)] $I_{H}(X; Y)\geq 0$.
\item[(DPI)] If $X-Y-Z$ forms a Markov chain, then 
\begin{align}
I_{H}(X; Z) \leq I_{H}(X; Y). \label{eq:DPI}
\end{align}
\end{description}
\end{theorem}

\begin{table*}[ht]
  \caption{Examples of $H$-mutual information ($H$-MI)}
  \label{tab:various_H-MI}
  \centering
  \resizebox{1.\textwidth}{!}{
  {
  \begin{tabular}{@{} ccccc @{}}
    \toprule
    Name of $H$-MI
    & $H(X)$
    & $\eta(t)$ 
    & $F(p_{X})$ 
    & $H(X | Y)$ 
    \\ 
    \midrule
    \begin{tabular}{c}
    Shannon MI \\ 
    $I(X; Y)$ \cite{shannon}  
    \end{tabular}
    & $-\sum_{x}p_{X}(x)\log p_{X}(x)$
    & $t$
    & $-\sum_{x}p_{X}(x)\log p_{X}(x)$
    & $-\sum_{y}p_{Y}(y)\sum_{x}p_{X\mid Y}(x | y)\log p_{X\mid Y}(x | y)$
    \\ 
    \begin{tabular}{c}
    Arimoto MI \\  
    $I_{\alpha}^{\text{A}}(X; Y)$ \cite{arimoto1977} 
    \end{tabular}
    & $\frac{\alpha}{1-\alpha}\log \norm{p_{X}}_{\alpha}$
    & $\begin{cases}
    \frac{\alpha}{1-\alpha} \log t, & 0<\alpha<1, \\
    \frac{\alpha}{1-\alpha} \log (-t), & \alpha>1
    \end{cases}$
    & $\begin{cases}
    \norm{p_{X}}_{\alpha}, & 0<\alpha<1, \\ 
    -\norm{p_{X}}_{\alpha}, & \alpha>1
    \end{cases}$
    & $\frac{\alpha}{1-\alpha}\log \sum_{y} p_{Y}(y)\sum_{x} \norm{p_{X\mid Y}(\cdot | y)}_{\alpha}$
    \\
    \begin{tabular}{c}
    Hayashi MI \\ 
     $I_{\alpha}^{\text{H}}(X; Y)$ 
    \end{tabular}
    & $\frac{1}{1-\alpha}\log \norm{p_{X}}_{\alpha}^{\alpha}$
    & $\begin{cases}
    \frac{1}{1-\alpha} \log t, & 0<\alpha<1, \\
    \frac{1}{1-\alpha} \log (-t), & \alpha>1
    \end{cases}$
    & $\begin{cases}
    \norm{p_{X}}_{\alpha}^{\alpha}, & 0<\alpha<1, \\ 
    -\norm{p_{X}}_{\alpha}^{\alpha}, & \alpha>1
    \end{cases}$
    & $\frac{1}{1-\alpha}\log \sum_{y} p_{Y}(y)\sum_{x} \norm{p_{X\mid Y}(\cdot | y)}_{\alpha}^{\alpha}$
    \\ 
    \begin{tabular}{c}
    Fehr--Berens MI \\
    $I_{\alpha}^{\text{FB}}(X; Y), \alpha>1$ 
    \end{tabular}
    & $-\log \norm{p_{X}}_{\alpha}^{\frac{\alpha}{\alpha-1}}$
    & $-\log (-t)$
    & $-\norm{p_{X}}_{\alpha}^{\frac{\alpha}{\alpha-1}}$
    & $-\log \sum_{y}p_{Y}(y) \norm{p_{X\mid Y}(\cdot |y)}_{\alpha}^{\frac{\alpha}{\alpha-1}}$
    \\ 
    \begin{tabular}{c}
    $\text{EVSI}^{(\cdot)}(X; Y)$ \\ 
    \cite{raiffa1961applied}, \cite{6483382}, \cite{ALVIM201932,6957119,6266165}
    \end{tabular}
    & \begin{tabular}{c}
    $\min_{q} \vE_{X}\left[\ell(X, q)\right]$, \\
    $-\max_{q} \vE_{X}\left[g(X, q)\right]$
    \end{tabular}
    & $t$
    & \begin{tabular}{c}
    $\min_{q} \vE_{X}\left[\ell(X, q)\right]$, \\
    $-\max_{q} \vE_{X}\left[g(X, q)\right]$
    \end{tabular}
    & \begin{tabular}{c}
    $\sum_{y}p_{Y}(y) \min_{q}\vE_{X}\left[\ell(X, q) \mid Y=y\right]$, \\
    $-\sum_{y}p_{Y}(y) \max_{q}\vE_{X}\left[g(X, q) \mid Y=y\right]$
    \end{tabular}
    \\
    \bottomrule
  \end{tabular}}
  }
\end{table*}


Table \ref{tab:various_H-MI} lists examples of $H$-MI, $H=(\eta, F)$, and $H(X | Y)$  described below that satisfy the conditions in Theorem \ref{thm:iff_DPI} (For more examples, see \cite{9505206,9064819}, \cite[Table I]{e24010039}). 

\begin{eg} \label{eg:H-MI}
Let $\alpha\in (0, 1)\cup (1, \infty)$. 
Shannon MI $I(X; Y):= {S(X) - S(X | Y)}$ and Arimoto MI $I_{\alpha}^{\text{A}}(X; Y):=H_{\alpha}(X)-H_{\alpha}^{\text{A}}(X\mid Y)$ are examples of $H$-MI, where
\begin{align}
H_{\alpha}(X) &:= \frac{\alpha}{1-\alpha}\log \norm{p_{X}}_{\alpha} = \frac{1}{1-\alpha}\log \norm{p_{X}}_{\alpha}^{\alpha} \\ 
&= -\log \norm{p_{X}}_{\alpha}^{\frac{\alpha}{\alpha-1}}, \\
H_{\alpha}^{\text{A}}(X\mid Y)&:= \frac{\alpha}{1-\alpha}\log \sum_{y} p_{Y}(y)\sum_{x} \norm{p_{X\mid Y}(\cdot \mid y)}_{\alpha}
\end{align}
are the R\`{e}nyi entropy of order $\alpha$ and the Arimoto conditional entropy of order $\alpha$ \cite{arimoto1977}, respectively.
\end{eg}

As shown in Example \ref{eg:H-MI}, the R\`{e}nyi entropy $H_{\alpha}(X)$ can be represented in at least three different ways. 
The corresponding $H=(\eta, F)$ for these expressions are shown in Table \ref{tab:various_H-MI}.
Thus, we can define novel MIs as follows:

\begin{definition}[Hayashi MI, Fehr--Berens MI]
Hayashi MI of order $\alpha\in (0, 1)\cup (1, \infty)$ and Fehr--Berens MI of order $\alpha>1$ are defined as follows:
\begin{align}
I_{\alpha}^{\text{H}}(X; Y) &:= H_{\alpha}(X) - H_{\alpha}^{\text{H}}(X \mid  Y), \\ 
I_{\alpha}^{\text{FB}}(X; Y) &:= H_{\alpha}(X) - H_{\alpha}^{\text{FB}}(X \mid  Y), 
\end{align}
where 
\begin{align}
H_{\alpha}^{\text{H}}(X; Y) &:= \frac{1}{1-\alpha}\log \sum_{y} p_{Y}(y)\sum_{x} \norm{p_{X\mid Y}(\cdot | y)}_{\alpha}^{\alpha}, \\
H_{\alpha}^{\text{FB}}(X; Y) &:= -\log \sum_{y}p_{Y}(y) \norm{p_{X\mid Y}(\cdot |y)}_{\alpha}^{\frac{\alpha}{\alpha-1}} 
\end{align}
are the Hayashi conditional entropy of order $\alpha$ \cite[Section II.A]{5773033} and the Fehr--Berens conditional entropy of order $\alpha$ \cite[Section III.E, 5)]{6898022}, respectively. 
\end{definition}

{Since $H_{\alpha}^{\text{A}}(X | Y)\geq H_{\alpha}^{\text{H}}(X | Y)$ \cite[Prop 1]{10.1007/978-3-319-04268-8_7}, it follows that Hayashi MI is greater than or equal to Arimoto MI. 
\begin{prop} \label{prop:Arimoto_less_than_Hayashi} 
Let $\alpha\in (0, 1)\cup (1, \infty)$.
\begin{align}
I_{\alpha}^{\text{A}}(X; Y) \leq I_{\alpha}^{\text{H}}(X; Y). 
\end{align}
\end{prop}
}


The amount of information that the observed data $Y$ contain about $X$ can also be quantified using the framework of a decision-making problem. 
In the 1960s, the EVSI was proposed by Raiffa and Schaifer \cite{raiffa1961applied}. 
Recently, equivalents or variants of the EVSI have been proposed in the context of privacy-guaranteed data-publishing problems.
For example, Calmon and Fawaz proposed average (cost) gain \cite{6483382} and Alvim \textit{et al.} proposed $g$-leakage \cite{6957119,ALVIM201932,6266165}.

\begin{definition} 
Let $g(x, a)$ be a gain function. The EVSI \cite{raiffa1961applied}, also known as \textit{average gain} \cite{6483382} and \textit{additive $g$-leakage} \cite{6957119,ALVIM201932,6266165}, 
is defined as the largest increase in the maximal Bayes expected gain compared to those without using $Y$, i.e., 
\begin{align}
&\text{EVSI}^{g}(X; Y) := \max_{\delta} G(\delta) - \max_{a} \vE_{X}\left[g(X, a)\right] \\ 
&= - \max_{a} \vE_{X}\left[g(X, a)\right] - \vE_{Y} \left[-\max_{a} \vE_{X}\left[g(X, a)\relmiddle{|} Y \right] \right], \label{eq:evsi_equality}
\end{align}
where the equality in \eqref{eq:evsi_equality} follows from Proposition \ref{prop:optimal_decision_rule}. 
The EVSI can also be defined using a loss function $\ell(x, a)$ as the largest reduction of the minimal Bayes risk compared with those without using $Y$, i.e., 
\begin{align}
&\text{EVSI}^{\ell}(X; Y) := \min_{a} \vE_{X}\left[\ell(X, a)\right] - \min_{\delta} r(\delta)\\ 
&= \max_{a} \vE_{X}\left[\ell(X, a)\right] - \vE_{Y} \left[\min_{a} \vE_{X}\left[\ell(X, a)\relmiddle{|} Y \right] \right]. 
\end{align}
\end{definition}

\begin{eg}
Suppose that a DM decides a pmf $q\in \Delta_{\mathcal{X}}$ considering log-loss $\ell_{\text{log}}(x, q):= -\log q(x)$ or log-score $g_{\text{log}}(x, q):=\log q(x)$. 
From Example \ref{eg:log_score}, we obtain 
\begin{align}
\text{EVSI}^{\ell_{\text{log}}}(X; Y) &= \text{EVSI}^{g_{\text{log}}}(X; Y) =  I(X; Y).
\end{align}
\end{eg}

Instead of examining the differences between $G(\delta)$ and $\vE_{X}[g(X, a)]$, one can quantify information leakage by examining their ratio.
Alvim \textit{et al.} proposed \textit{multiplicative} $g$-leakage \cite{6957119,ALVIM201932,6266165} as follows:
\begin{definition}[multiplicative $g$-leakage]
\footnote{We slightly modified the definition of the multiplicative $g$-leakage so that we can define it using non-positive gain function $g(x,a)$ by multiplying $c(g)$.}
Let $g(x, a)$ be a non-negative or non-positive gain function and $c(g)$ be a function of $g$ such that its sign is equal to $\text{sign}(g)$\footnote{$\text{sign}(g):=1$, if $g(x, a)\geq 0, \forall (x, a)$, $-1$; otherwise. }. 
Then the \textit{multiplicative $g$-leakage} is defined as the largest multiplicative increase of the maximal Bayes expected gain compared to those of without $Y$, i.e., 
\begin{align}
&\text{MEVSI}^{g}(X; Y) := c(g) \log \frac{\max_{\delta} G(\delta)}{\max_{a} \vE_{X}\left[g(X, a)\right]} \\ 
&= c(g) \log \frac{\vE_{Y} \left[\max_{a} \vE_{X}\left[g(X, a)\relmiddle{|} Y \right] \right]}{\max_{a} \vE_{X}\left[g(X, a)\right]}.
\end{align}
Similarly, we can define $\text{MEVSI}^{\ell}(X; Y)$ using a loss function $\ell(x, a)$. 
\end{definition}

\begin{eg} \label{eg:Arimoto_MI_expression}
Suppose that a DM decides a pmf $q\in \Delta_{\mathcal{X}}$ considering 
pseudo-spherical score $g_{\text{PS}}(x, q):= \frac{1}{\alpha-1}\cdot \left(\frac{q(x)}{\norm{q}_{\alpha}} \right)^{\alpha-1}$ 
or $g_{\alpha}(x, q):=\frac{\alpha}{\alpha-1}\cdot q(x)^{\frac{\alpha-1}{\alpha}}$ (referred to as $\alpha$-score). 
Define $c(g_{\text{PS}}) = c(g_{\alpha}) := \frac{\alpha}{\alpha-1}$. From Table \ref{tab:tab_scoring_rule}, we obtain 
\begin{align}
\text{MEVSI}^{g_{\text{PS}}}(X; Y) &= \text{MEVSI}^{g_{\alpha}}(X; Y)  \notag \\
&= I_{\alpha}^{\text{A}}(X; Y). \qquad \text{(Arimoto MI)}
\end{align}
\end{eg}

\begin{eg}
Suppose that a DM decides a pmf $q\in \Delta_{\mathcal{X}}$ considering 
a power score $g_{\text{Power}}(x, q):=\frac{\alpha}{\alpha-1}\cdot q(x)^{\alpha-1} - \norm{q}_{\alpha}^{\alpha}$.
Define $c(g_{\text{Power}}) := \frac{1}{\alpha-1}$. From Table \ref{tab:tab_scoring_rule}, we obtain 
\begin{align}
\text{MEVSI}^{g_{\text{Power}}}(X; Y) &= I_{\alpha}^{\text{H}}(X; Y). \qquad \text{(Hayashi MI)} 
\end{align}
\end{eg}

Note that we can easily show that $F(p_{X}) := -\vE_{X}^{p_{X}}[g(X, a)]$ and $F(p_{X}):=\vE_{X}^{p_{X}}[\ell(X, a)]$ are concave with respect to $p_{X}$ 
and that $H(X | Y):= \vE_{Y} \left[-\max_{a} \vE_{X}\left[g(X, a)\relmiddle{|} Y \right]\right]$ and $H(X | Y):=$\\$\vE_{Y} \left[\min_{a} \vE_{X}\left[\ell(X, a)\relmiddle{|} Y \right] \right]$ satisfy the EAVG condition given in Definition \ref{def:EAVG} (see also \cite[Sec V.F]{9064819}).
Thus, we obtain the following result.
\begin{prop}[\text{\cite[Sec V.F]{9064819}}]
$\text{EVSI}^{(\cdot)}(X; Y)$ and $\text{MEVSI}^{(\cdot)}$ are members of $H$-MI. 
\end{prop}

Conversely, can we represent $H$-MI $I_{H}(X; Y)$ by a decision-theoretic quantity?
In the next section, we will show that this is possible. Furthermore, we derive a variational characterization of $H$-MI using this representation.

\section{Variational Characterization of $H$-MI}\label{sec:variational_characterization_H-MI}
In this section, we provide a variational characterization of $H$-MI $I_{H}(X; Y)$ using the fact that every continuous concave function $F$ has a statistical decision-theoretic variational characterization \cite[Section 3.5.4]{grunwald2004}. 

Gr{\"u}nwald and Dawid showed that every concave function $F\colon \Delta_{\mathcal{X}}\to \br$ has the following variational characterization. 

\begin{prop} [\text{\cite[Section 3.5.4]{grunwald2004}}] \label{prop:variational_characterization_F}
Let $\mathcal{X}=\{x_{1},x_{2}, \dots, x_{m}\}$ and $F\colon \Delta_{\mathcal{X}}\to \br$ be a continuous concave functions. 
Suppose that a DM decide a pmf $q\in \Delta_{\mathcal{X}}\subseteq [0, 1]^{m}$ considering 
the following proper loss function $\ell_{F}(x, q)$ defined as 
\begin{align}
\ell_{F}(x, q) := F(q) + z^{\top} (\mbox{1}\hspace{-0.25em}\mbox{l}^{x}-q),  \label{eq:l_F}
\end{align}
where 
\begin{itemize}
\item $\mbox{1}\hspace{-0.25em}\mbox{l}^{x}$ is the $m$-dimensional vector having $\mbox{1}\hspace{-0.25em}\mbox{l}^{x}_{j}=1$ if $j=x$, $0$; otherwise, 
\item $z\in \partial F(q)\subseteq \br^{m}$ is a subgradient in subdifferential of $F(q)$\footnote{Note that if $F$ is {differentiable}, then the subdifferential $\partial F(q)$ is singleton, i.e., $\partial F(q)=\{\nabla F(q)\}$, where $\nabla F(q)$ is the gradient of $F(q)$. }.
\end{itemize}
Then, the following holds:
\begin{align}
F(p_{X}) &= \min_{q} \vE_{X}^{p_{X}}\left[\ell_{F}(X, q)\right], 
\end{align}
where the minimum is achieved at $q=p_{X}$.
\end{prop}

\begin{eg} \label{eg:examples_l_F}
Some examples of the proper loss function $\ell_{F}(x, q)$ in Proposition \ref{prop:variational_characterization_F} are listed below: 
\begin{itemize}
\item If $F(p_{X}) = -\sum_{x}p_{X}(x)\log p_{X}(x)$, then $\ell_{F}(x, q) = \ell_{\text{log}}(x, q) = -g_{\text{log}}(x, q) = -\log q(x)$. 
\item If $F(p_{X}) = \norm{p_{X}}_{\alpha}, 0<\alpha<1$, then $\ell_{F}(x, q) = \left(\frac{q(x)}{\norm{q}_{\alpha}} \right)^{\alpha-1}=(\alpha-1)g_{\text{PS}}(x, q)$. 
If $F(p_{X}) = -\norm{p_{X}}_{\alpha}, \alpha>1$, then $\ell_{F}(x, q) = (1-\alpha)g_{\text{PS}}(x, q)$. 
\item If $F(p_{X}) = \norm{p_{X}}_{\alpha}^{\alpha}, 0<\alpha<1$, then $\ell_{F}(x, q) = \alpha q(x)^{\alpha-1} - (\alpha-1)\norm{q}_{\alpha}^{\alpha} = (\alpha-1) g_{\text{Power}}(x, q)$. 
If $F(p_{X}) = -\norm{p_{X}}_{\alpha}^{\alpha}, \alpha>1$, then $\ell_{F}(x, q) = (1-\alpha) g_{\text{Power}}(x, q)$.
\item If $F(p_{X}) = -\norm{p_{X}}_{\alpha}^{\frac{\alpha}{\alpha-1}}, \alpha>1$, then $\ell_{F}(x, q) = \norm{q}_{\alpha}^{\alpha-1} - \frac{\alpha}{\alpha-1}(\norm{q}_{\alpha}^{\alpha} - q(x)^{\alpha-1})$. 
\end{itemize}
\end{eg}

Using Proposition \ref{prop:variational_characterization_F}, we obtain the following variational characterization of $H$-MI. 

\begin{theorem}[Variational characterization of $H$-MI] \label{thm:variational_characterization}
Suppose that $H=(\eta, F)$ satisfies the CCV condition and $H(X\mid Y)$ satisfies the EAVG condition, respectively. 
Then, there exists a functional ${\mathcal{F}_{H}}(p_{X}, q_{X\mid Y})$ such that 
\begin{align}
I_{H}(X; Y) &= \max_{q_{X\mid Y}}{\mathcal{F}_{H}}(p_{X}, q_{X\mid Y}).
\end{align}
\end{theorem}

\begin{proof}
From Proposition \ref{prop:variational_characterization_F}, there exists a proper loss function $\ell_{F}(x, q)$ such that $F(p_{X}) = \min_{q}\vE_{X}^{p_{X}}\left[\ell_{F}(X, q)\right]$. 
Since $H(X | Y)$ satisfies EAVG, it can be written as 
\begin{align}
H(X\mid Y) &= \eta\left( \vE_{Y}\left[F(p_{X\mid Y}(\cdot\mid Y))\right] \right) \\ 
&= \eta\left( \vE_{Y}\left[\min_{q}\vE_{X}^{p_{X\mid Y}(\cdot \mid Y)}\left[\ell_{F}(X, q)\right]\right] \right) \\ 
&= \eta\left( \vE_{Y}\left[\min_{q}\vE_{X}\left[\ell_{F}(X, q) \mid Y\right]\right] \right) \\ 
&\overset{(a)}{=} \eta\left( \min_{q_{X\mid Y}} \vE_{X, Y}\left[\ell_{F}(X, q_{X\mid Y}(X \mid Y))\right]\right) \\ 
&\overset{(b)}{=} \min_{q_{X\mid Y}}\eta \left( \vE_{X, Y}\left[\ell_{F}(X, q_{X\mid Y}(X \mid Y))\right] \right), 
\end{align}
where 
\begin{itemize}
\item $(a)$ follows from Proposition \ref{prop:optimal_decision_rule} and Remark  \ref{rem:q_X_Y_risk}, 
\item $(b)$ follows from the assumption that $\eta$ is strictly increasing.
\end{itemize}
Therefore, we obtain the following variational characterization of $H$-MI:
\begin{align}
&I_{H}(X; Y) := \eta(F(p_{X})) - \eta\left( \vE_{Y}\left[F(p_{X\mid Y}(X\mid Y))\right] \right) \\ 
&= \eta(F(p_{X})) - \min_{q_{X\mid Y}}\eta \left( \vE_{X, Y}\left[\ell_{F}(X, q_{X\mid Y}(X \mid Y))\right] \right) \\ 
&= \max_{q_{X\mid Y}} \underbrace{\left( \eta(F(p_{X})) -\eta \left( \vE_{X, Y}\left[\ell_{F}(X, q_{X\mid Y}(X \mid Y))\right] \right) \right)}_{=:{\mathcal{F}_{H}}(p_{X}, q_{X\mid Y})}.
\end{align}
\end{proof}
\begin{eg} 
From Theorem \ref{thm:variational_characterization} and Example \ref{eg:examples_l_F} we obtain the variational characterization for specific $H$-MIs as follows: 
\begin{align}
&I(X; Y) 
= \max_{q_{X\mid Y}}\vE_{X, Y}^{p_{X}p_{Y\mid X}}\left[\log \frac{q_{X\mid Y}(X\mid Y)}{p_{X}(X)}\right],  \label{eq:variational_characterization_Shannon_MI} \\ 
&I_{\alpha}^{\text{A}}(X; Y) = \max_{q_{X\mid Y}} \frac{\alpha}{\alpha-1} \log \frac{\vE_{X, Y}^{p_{X}p_{Y\mid X}}\left[\left( \frac{q_{X\mid Y}(X\mid Y)}{\norm{q_{X\mid Y}(\cdot \mid Y)}_{\alpha}} \right)^{\alpha-1}\right]}{\norm{p_{X}}_{\alpha}}, \label{eq:variational_characterization_Arimoto_MI2} \\ 
&I_{\alpha}^{\text{H}}(X; Y) = \max_{q_{X\mid Y}}\frac{1}{\alpha-1} \times\notag \\
& \log \frac{\vE_{X, Y}^{p_{X}p_{Y\mid X}}\left[\alpha q_{X\mid Y}(X\mid Y)^{\alpha-1} - (\alpha-1)\norm{q_{X\mid Y}(\cdot\mid Y)}_{\alpha}^{\alpha}\right]}{\norm{p_{X}}_{\alpha}^{\alpha}}, \label{eq:variational_characterization_Hayashi_MI} \\ 
&I_{\alpha}^{\text{FB}}(X; Y) = \max_{q_{X\mid Y}} \log \frac{\vE_{X, Y}^{p_{X}p_{Y\mid X}}\left[\ell^{\text{FB}}(X, q_{X\mid Y}(X\mid Y))\right]}{\norm{p_{X}}_{\alpha}^{\frac{\alpha}{\alpha-1}}}, \label{eq:variational_characterization_FB_MI}
\end{align}
where $\ell^{\text{FB}}(x, q):=\norm{q}_{\alpha}^{\alpha-1} - \frac{\alpha}{\alpha-1}(\norm{q}_{\alpha}^{\alpha} - q(x)^{\alpha-1})$. 
\end{eg}

\begin{remark}
From Example \ref{eg:Arimoto_MI_expression}, we obtain another variational characterization with $\ell_{F}(x, q)=-g_{\alpha}(x, q)$ that is \textit{not} proper as follows:
\begin{align}
I_{\alpha}^{\text{A}}(X; Y) &= \max_{q_{X\mid Y}} \frac{\alpha}{\alpha-1} \log \frac{\vE_{X, Y}^{p_{X}p_{Y\mid X}}\left[q_{X\mid Y}(X\mid Y)^{\frac{\alpha-1}{\alpha}}\right]}{\norm{p_{X}}_{\alpha}}. \label{eq:variational_characterization_Arimoto_MI1}
\end{align}
\end{remark}


\section{Application: Deriving Algorithm For Computing $H$-Capacity}\label{sec:application_ABA}
In information theory, the notion of capacity often characterizes the theoretical limits of performance in the problem. 
For example, channel capacity $C:=\max_{p_{X}}I(X; Y)$ characterizes supremum of achievable rate in channel coding\cite{shannon}. 
Recently, Liao \textit{et al.} reported the operational meaning of Arimoto capacity $C_{\alpha}^{\text{A}}:=\max_{p_{X}}I_{\alpha}^{\text{A}}(X; Y)$ 
in the privacy-guaranteed data-publishing problems \cite[Thm 2]{8804205}. 
The Arimoto--Blahut algorithm (ABA), which is a well-known alternating optimization algorithm for computing capacity $C$, proposed by Arimoto \cite{1054753} and Blahut \cite{1054855}.
Extending his results, Arimoto derived an ABA for computing Arimoto capacity $C_{\alpha}^{\text{A}}$  in \cite{BN01990060en}. 
Recently, we derived another ABA for computing $C_{\alpha}^{\text{A}}$ using a variational characterization of $I_{\alpha}^{\text{A}}(X; Y)$ different from Arimoto's method \cite{kamatsuka2024new}. 
These algorithms are based on a double maximization problem using the variational characterization of MIs. 
In this section, we derive an alternating optimization algorithm for computing $H$-capacity $C_{H}:=\max_{p_{X}}I_{H}(X; Y)$ based on the variational characterization of $H$-MI and ABA.
Moreover, we show that the algorithms for computing Arimoto capacity $C_{\alpha}^{\text{A}}$ from our approach coincide with the previous algorithms \cite{BN01990060en}, \cite{kamatsuka2024new}. 

From Theorem \ref{thm:variational_characterization}, $H$-capacity $C_{H}:=\max_{p_{X}}I_{H}(X; Y)$ can be represented as a double maximization problem as follows:
\begin{align}
C_{H} &= \max_{p_{X}} \max_{q_{X\mid Y}} {\mathcal{F}_{H}}(p_{X}, q_{X\mid Y}), \label{eq:double_maximum}
\end{align}
where 
\begin{align}
&{\mathcal{F}_{H}}(p_{X}, q_{X\mid Y})\notag \\ 
&:=\left( \eta(F(p_{X})) -\eta \left( \vE_{X, Y}\left[\ell_{F}(X, q_{X\mid Y}(X \mid Y))\right] \right) \right).
\end{align}

Based on the representation in \eqref{eq:double_maximum}, we can derive an alternating optimization algorithm for computing $C_{H}$ as 
described in Algorithm \ref{alg:ABA_H}, where $p_{X}^{(0)}$ is an initial distribution of the algorithm.. 

\begin{algorithm}[h]
	\caption{Arimoto--Blahut algorithm for computing $C_{H}$}
	\label{alg:ABA_H}
	\begin{algorithmic}[1]
		\Require 
			\Statex $p_{X}^{(0)}, p_{Y\mid X}$, $\epsilon\in (0, 1)$
		\Ensure
			\Statex { approximation of} $C_{H}$
		\Initialize{
			$q_{X\mid Y}^{(0)} \gets \argmax_{q_{X\mid Y}}{\mathcal{F}_{H}}(p_{X}^{(0)}, q_{X\mid Y})$\\ 
			$F^{(0, 0)}\gets {\mathcal{F}_{H}}(p_{X}^{(0)}, q_{X\mid Y}^{(0)})$ \\ 
			$k\gets 0$ \\
      }
		\Repeat
			\State $k\gets k+1$
			\State $p_{X}^{(k)} \gets \argmax_{p_{X}}{\mathcal{F}_{H}}(p_{X}, q_{X\mid Y}^{(k-1)})$

			\State $q_{X\mid Y}^{(k)} \gets \argmax_{q_{X\mid Y}}{\mathcal{F}_{H}}(p_{X}^{(k)}, q_{X\mid Y})$
			\State $F^{(k, k)} \gets {\mathcal{F}_{H}}(p_{X}^{(k)}, q_{X\mid Y}^{(k)})$
		\Until{$\abs{F^{(k, k)} - F^{(k-1, k-1)}} < \epsilon$} 
		\State \textbf{return} $F^{(k, k)}$
	\end{algorithmic}
\end{algorithm}

\begin{table*}[ht]
  \caption{Formulae for updating $p_{X}^{(k)}$ and $q_{X\mid Y}^{(k)}$ in the Arimoto--Blahut Algorithm for calculating $H$-capacity $C_{H}$ (cited from \cite[Table I]{kamatsuka2024new})}
  \label{tab:update_formula_H}
  \resizebox{1.\textwidth}{!}{
  \centering
  \begin{tabular}{@{} cccc @{}}
    \toprule
    Name & ${\mathcal{F}_{H}}(p_{X}, q_{X\mid Y})$ & $p_{X}^{(k)}$ & $q_{X\mid Y}^{(k)}$  \\ 
    \midrule
    \begin{tabular}{c}
    ABA for \\ 
    computing $C$ \cite{1054753}, \cite{1054855}
    \end{tabular}
    & \begin{tabular}{c}
    $\vE_{X, Y}^{p_{X}p_{Y\mid X}} \left[\log \frac{q_{X\mid Y}(X\mid Y)}{p_{X}(X)}\right]$ 
    \end{tabular}
    & $\frac{\prod_{y} q_{X\mid Y}^{(k-1)}(x | y)^{p_{Y\mid X}(y | x)}}{\sum_{x}\prod_{y} q_{X\mid Y}^{(k-1)}(x | y)^{p_{Y\mid X}(y | x)}}$ 
    & $\frac{p_{X}^{(k)}(x)p_{Y\mid X}(y | x)}{\sum_{x}p_{X}^{(k)}(x) p_{Y\mid X}(y | x)}$  \\ 
    \midrule
    \begin{tabular}{c}
    ABA for \\
    computing $C_{\alpha}^{\text{A}}$ \cite{BN01990060en}
    \end{tabular}
    & \begin{tabular}{c}
    $\frac{\alpha}{\alpha-1} \log \sum_{x, y} p_{X_{\alpha}}(x)^{\frac{1}{\alpha}}p_{Y\mid X}(y | x) q_{X\mid Y}(x | y)^{\frac{\alpha-1}{\alpha}}$ 
    \end{tabular}
    & $\frac{\left( \sum_{y}p_{Y\mid X}(y | x)q_{X\mid Y}^{(k-1)}(x | y)^{\frac{\alpha-1}{\alpha}} \right)^{\frac{1}{\alpha-1}}}{\sum_{x}\left( \sum_{y}p_{Y\mid X}(y | x)q_{X\mid Y}^{(k-1)}(x | y)^{\frac{\alpha-1}{\alpha}} \right)^{\frac{1}{\alpha-1}}}$ 
    & $\frac{p_{X}^{(k)}(x)^{\alpha}p_{Y\mid X}(y | x)^{\alpha}}{\sum_{x}p_{X}^{(k)}(x)^{\alpha}p_{Y\mid X}(y | x)^{\alpha}}$  \\ 
    \begin{tabular}{c}
    ABA for \\
    computing $C_{\alpha}^{\text{A}}$ \cite{kamatsuka2024new}
    \end{tabular}
    & \begin{tabular}{c}
    $\frac{\alpha}{\alpha-1} \log \sum_{x, y} p_{X_{\alpha}}(x)^{\frac{1}{\alpha}}p_{Y\mid X}(y | x) q_{X_{\alpha}\mid Y}(x | y)^{\frac{\alpha-1}{\alpha}}$ 
    \end{tabular}
    & $\frac{\left( \sum_{y}p_{Y\mid X}(y | x) q_{X_{\alpha}\mid Y}^{(k-1)}(x | y)^{\frac{\alpha-1}{\alpha}} \right)^{\frac{1}{\alpha-1}}}{\sum_{x}\left( \sum_{y}p_{Y\mid X}(y | x) q_{X_{\alpha}\mid Y}^{(k-1)}(x | y)^{\frac{\alpha-1}{\alpha}} \right)^{\frac{1}{\alpha-1}}}$ 
    & $\frac{p_{X}^{(k)}(x)p_{Y\mid X}(y | x)}{\sum_{x}p_{X}^{(k)}(x)p_{Y\mid X}(y | x)}$  \\ 
    \bottomrule
  \end{tabular}
  }
\end{table*}

From Propositions \ref{prop:optimal_decision_rule} and \ref{prop:variational_characterization_F}, the optimum $q_{X\mid Y}^{*}=\argmax_{q_{X\mid Y}}{\mathcal{F}_{H}}(p_{X}, q_{X\mid Y})$
 for a fixed $p_{X}$ is obtained as follows.
\begin{prop}
For a fixed $p_{X}$, ${\mathcal{F}_{H}}(p_{X}, q_{X\mid Y})$ is maximized by 
\begin{align}
q_{X\mid Y}^{*}(x\mid y) &= p_{X\mid Y}(x\mid y) = \frac{p_{X}(x)p_{Y\mid X}(y\mid x)}{\sum_{x}p_{X}(x)p_{Y\mid X}(y\mid x)}.
\end{align}
\end{prop}

\begin{proof}
It can be easily checked that finding the optimum $q_{X\mid Y}^{*}=\argmax_{q_{X\mid Y}}{\mathcal{F}_{H}}(p_{X}, q_{X\mid Y})$ for fixed $p_{X}$ is equivalent to finding the optimum $q_{X\mid Y}^{*}=\argmin_{q_{X\mid Y}}\vE_{X, Y}\left[\ell_{F}(X, q_{X\mid Y}(X\mid Y))\right]$. 
From Proposition \ref{prop:optimal_decision_rule}, the problem of finding $q_{X\mid Y}=\{q_{X\mid Y}(\cdot | y)\}_{y\in \mathcal{Y}}$ that minimizes $\vE_{X, Y}\left[\ell_{F}(X, q_{X\mid Y}(X\mid Y))\right]$ 
becomes equivalent to the problem of finding the optimal conditional distribution $q_{X\mid Y}(\cdot | y)$ for each $y\in \mathcal{Y}$ that 
minimizes $\vE_{X}\left[\ell(X, q_{X\mid Y}(\cdot | y))\mid Y=y\right] = \vE_{X}^{p_{X\mid Y}(\cdot \mid y)}\left[\ell(X, q_{X\mid Y}(\cdot | y))\right]$.
Since $\ell_{F}(x, q)$ defined in \eqref{eq:l_F} is proper, the optimum is obtained as $q_{X\mid Y}^{*}(\cdot | y) = p_{X\mid Y}(\cdot | y), y\in \mathcal{Y}$. 
\end{proof}

\begin{remark} \label{rem:depends}
On the other hand, whether the optimum $p_{X}^{*}=\argmax_{p_{X}}{\mathcal{F}_{H}}(p_{X}, q_{X\mid Y})$
for a fixed $q_{X\mid Y}$ can be obtained explicitly depends on $H=(\eta, F)$. 
For example, Arimoto \cite{1054753} and Blahut \cite{1054855} derived the explicit formula for $p_{X}^{*}$, where ${\mathcal{F}}(p_{X}, q_{X\mid Y}):=\vE_{X, Y}^{p_{X}p_{Y\mid X}}\left[\log \frac{q_{X\mid Y}(X\mid Y)}{p_{X}(X)}\right]$ is defined in \eqref{eq:variational_characterization_Shannon_MI}. 
Table \ref{tab:update_formula_H} lists the explicit updating formulae for computing channel capacity $C$.
However, when computing Hayashi capacity $C_{\alpha}^{\text{H}}:=\max_{p_{X}}I_{\alpha}^{\text{H}}(X; Y)$ and 
Fehr--Berens capacity $C_{\alpha}^{\text{FB}}:=\max_{p_{X}}I_{\alpha}^{\text{FB}}(X; Y)$, it seems that there is no explicit updating formula for $p_{X}^{*}$ for a fixed $q_{X\mid Y}$. 
Therefore, one must find it numerically.
\end{remark}

Next, we consider driving the algorithms for computing the Arimoto capacity $C_{\alpha}^{\text{A}}$. 
Based on the variational characterizations \eqref{eq:variational_characterization_Arimoto_MI1} and \eqref{eq:variational_characterization_Arimoto_MI2}, 
we define functionals ${\mathcal{F}_{\alpha}}^{\text{A1}}(p_{X}, q_{X\mid Y})$ and ${\mathcal{F}_{\alpha}}^{\text{A2}}(p_{X}, q_{X\mid Y})$ as follows:
\begin{align}
{\mathcal{F}_{\alpha}}^{\text{A1}}(p_{X}, q_{X\mid Y}) &:= \frac{\alpha}{\alpha-1} \log \frac{\vE_{X, Y}^{p_{X}p_{Y\mid X}}\left[q_{X\mid Y}(X\mid Y)^{\frac{\alpha-1}{\alpha}}\right]}{\norm{p_{X}}_{\alpha}}, \\ 
{\mathcal{F}_{\alpha}}^{\text{A2}}(p_{X}, q_{X\mid Y}) &:= \frac{\alpha}{\alpha-1} \log \frac{\vE_{X, Y}^{p_{X}p_{Y\mid X}}\left[\left( \frac{q_{X\mid Y}(X\mid Y)}{\norm{q_{X\mid Y}(\cdot \mid Y)}_{\alpha}} \right)^{\alpha-1}\right]}{\norm{p_{X}}_{\alpha}}.
\end{align}

Simple calculations yield the following result.

\begin{prop}
\begin{align}
&{\mathcal{F}_{\alpha}}^{\text{A1}}(p_{X}, q_{X\mid Y}) \notag \\ 
&= \frac{\alpha}{\alpha-1} \log \sum_{x, y} p_{X_{\alpha}}(x)^{\frac{1}{\alpha}} p_{Y\mid X}(y | x) q_{X\mid Y}(x | y)^{\frac{\alpha-1}{\alpha}}, \\
&{\mathcal{F}_{\alpha}}^{\text{A2}}(p_{X}, q_{X\mid Y}) \notag \\
&= \frac{\alpha}{\alpha-1} \log \sum_{x, y} p_{X_{\alpha}}(x)^{\frac{1}{\alpha}} p_{Y\mid X}(y | x) q_{X_{\alpha}\mid Y}(x | y)^{\frac{\alpha-1}{\alpha}}, 
\end{align}
where $p_{X_{\alpha}}$ is the $\alpha$-tilted distribution of $p_{X}$ defined in \eqref{eq:alpha_tilted_dist} and $q_{X_{\alpha}\mid Y}=\{q_{X_{\alpha}\mid Y}(\cdot | y)\}_{y\in \mathcal{Y}}$ 
is a set of $\alpha$-tilted distribution of $q_{X\mid Y}(\cdot | y)$ defined as $q_{X_{\alpha}\mid Y}(x | y):=\frac{q_{X\mid Y}(x|y)^{\alpha}}{\sum_{x}q_{X\mid Y}(x|y)^{\alpha}}$.  
\end{prop}

The variational characterization $I_{\alpha}^{\text{A}}(X; Y) =$\\$ \max_{q_{X\mid Y}} {\mathcal{F}_{\alpha}}^{\text{A1}}(p_{X}, q_{X\mid Y})$ is equivalent to 
that presented in \cite[Eq. (7.103)]{BN01990060en} by Arimoto (see also \cite[Prop 4 and Remark 4]{kamatsuka2024new}). 
On the other hand, the variational characterization $I_{\alpha}^{\text{A}}(X; Y) = \max_{q_{X\mid Y}} {\mathcal{F}_{\alpha}}^{\text{A2}}(p_{X}, q_{X\mid Y})$ is equivalent to that presented in \cite[Thm 1]{kamatsuka2024new}. 
Therefore, Algorithm \ref{alg:ABA_H} applied for computing the Arimoto capacity $C_{\alpha}^{\text{A}}$ is equivalent to those  previously presented in \cite{BN01990060en}, \cite{kamatsuka2024new}. 
Table \ref{tab:update_formula_H} lists the explicit updating formulae for computing Arimoto  capacity $C_{\alpha}^{\text{A}}$ of each algorithm.

{

Finally, we discuss the global convergence property of Algorithm \ref{alg:ABA_H}. 
In general, there is no guarantee that Algorithm \ref{alg:ABA_H} exhibits global convergence property, and whether it does or not depends on the given $H=(\eta, F)$.
However, the following sufficient condition on $H=(\eta, F)$ for the global convergence can be immediately obtained from \cite[Thm 10.5]{10.5555/1199866}.
\begin{prop} 
Let $\{p_{X}^{(k)}\}_{k=0}^{\infty}$ and $\{q_{X\mid Y}^{(k)}\}_{k=0}^{\infty}$ be sequences of distributions obtained from Algorithm \ref{alg:ABA_H}. 
If $(p_{X}, q_{X\mid Y}) \mapsto \mathcal{F}_{H}(p_{X}, q_{X\mid Y})$ is jointly concave, then 
\begin{align}
\lim_{k\to \infty} \mathcal{F}_{H}(p_{X}^{(k)}, q_{X\mid Y}^{(k)}) = C_{H}.
\end{align}
\end{prop}

\begin{remark}
$\mathcal{F}(p_{X}, q_{X\mid Y}) :=\vE_{X, Y}^{p_{X}p_{Y\mid X}}\left[\log \frac{q_{X\mid Y}(X\mid Y)}{p_{X}(X)}\right]$ is a typical example that satisfies this condition (see \cite[Section 10.3.2]{10.5555/1199866}).
Note that even if $H=(\eta, F)$ does not satisfy this sufficient condition, it may be possible to show the global convergence property of Algorithm \ref{alg:ABA_H}. 
For example, Kamatsuka \textit{et al}. \cite[Cor 2]{kamatsuka2024new} proved that 
\begin{align}
\lim_{k\to \infty} \mathcal{F}_{\alpha}^{\text{A1}}(p_{X}^{(k)}, q_{X\mid Y}^{(k)}) &= \lim_{k\to \infty} \mathcal{F}_{\alpha}^{\text{A2}}(p_{X}^{(k)}, q_{X\mid Y}^{(k)}) = C_{\alpha}^{\text{A}}
\end{align}
by showing the equivalence of the proposed algorithm with the alternating optimization algorithm for which global convergence is guaranteed by Arimoto \cite[Thm 3]{1055640}.
\end{remark}

}

\section{Conclusion}\label{sec:conclusion}
In this study, we derived a variational characterization of $H$-MI $I_{H}(X; Y)$. 
On the basis of the characterization, we derived an alternating optimization algorithm for $H$-capacity $C_{H} := \max_{p_{X}} I_{H}(X; Y)$. 
We also showed that the algorithms applied for computing Arimoto capacity $C_{\alpha}^{\text{A}}$ coincide with the previously reported algorithms \cite{BN01990060en}, \cite{kamatsuka2024new}.
In a future study, we will derive algorithms for the calculating Hayashi capacity $C_{\alpha}^{\text{H}}:=\max_{p_{X}}I_{\alpha}^{\text{H}}(X; Y)$ and 
Fehr--Berens capacity $C_{\alpha}^{\text{FB}}:=\max_{p_{X}}I_{\alpha}^{\text{FB}}(X; Y)$. 

\section*{Acknowledgments}
This work was supported by JSPS KAKENHI Grant Number JP23K16886. 


\end{document}